\documentclass[11pt, a4paper]{article}

\usepackage[a4paper,margin=2.5cm]{geometry}
\usepackage[utf8]{inputenc}
\usepackage{booktabs}
\usepackage[small]{caption}
\usepackage{lmodern}

\usepackage{amsmath, amsthm, amsfonts, amssymb, mathtools,booktabs,placeins, dsfont}

\usepackage{xspace}

\usepackage{tikz}
\usepackage{aligned-overset}

\usepackage{hyperref}
\usepackage{thm-restate}
\usepackage{enumerate}

\usepackage{bm}

\usepackage{todonotes}

\usepackage[sort&compress,numbers]{natbib}
\usepackage{cleveref}
\theoremstyle{plain}
\newtheorem{theorem}{Theorem}[section]
\newtheorem{definition}[theorem]{Definition}
\newtheorem{lemma}[theorem]{Lemma}

\newtheorem{proposition}[theorem]{Proposition}

\theoremstyle{remark}
\newtheorem*{remark}{Remark}

\newtheorem{example}[theorem]{Example}
\newtheorem{observation}[theorem]{Observation}

\definecolor{blue}{HTML}{0072B2}
\definecolor{lightblue}{HTML}{56B4E9}
\definecolor{rose}{HTML}{CC79A7}
\definecolor{orange}{HTML}{D55E00}
\definecolor{yellow}{HTML}{F0E442}
\definecolor{green}{HTML}{009E73}

\newcommand{\ALG}{\textsc{Alg}\xspace}
\newcommand{\OPT}{\textsc{Opt}\xspace}
\newcommand{\FIFO}{\textrm{FIFO}\xspace}
\newcommand{\LRU}{\textrm{LRU}\xspace}
\newcommand{\rank}[1]{\operatorname{rank}(#1)}

\newcommand{\pageUniverse}{V}
\newcommand{\numPages}{n}
\newcommand{\cache}[1]{C_{#1}}
\newcommand{\initialCache}{\cache{0}}

\newcommand{\requestPage}[1]{v_{#1}}
\newcommand{\arrival}[1]{a_{#1}}
\newcommand{\deadline}[1]{d_{#1}}
\newcommand{\delay}[1]{\operatorname{delay}_{#1}}
\newcommand{\serviceTime}[1]{s_{#1}}

\newcommand{\thresholdValue}{\delta}
\newcommand{\paymentCount}{P}
\newcommand{\paymentPage}[1]{v_{#1}}
\newcommand{\fixedPage}[1]{f_{#1}}
\newcommand{\paymentTime}[1]{t_{#1}}
\newcommand{\paymentRequests}[1]{R_{#1}}
\newcommand{\paymentWindow}[1]{W_{#1}}

\newcommand{\offlineFetchCost}{F_\OPT}
\newcommand{\offlineDelayCost}{D_\OPT}
\newcommand{\potential}{\Phi}
\newcommand{\potentialAt}[1]{\potential_{#1}}
\newcommand{\finalPotential}{\potentialAt{\mathrm{final}}}

\newcommand{\candidates}{\mathcal{C}}
\newcommand{\offlineBudget}{B}

\renewcommand{\O}{\mathcal O}

\usetikzlibrary{arrows.meta,positioning,shapes.misc, decorations.pathreplacing}

\tikzset{
    alg/.style={draw=blue, very thick, yshift=2pt},
    opt/.style={draw=orange, very thick, yshift=-2pt},
    release/.style={solid, circle, draw=black, fill=white, inner sep=0.5ex, thick},
    threshold/.style={solid, draw=black, shape=cross out, minimum size=1.4ex, inner sep=0pt, thick}
}

\newcommand{\mytitle}[1]{
    \begingroup
    \fontsize{16pt}{18pt}\fontseries{bx}\selectfont \centering #1 \par  
    \endgroup
}
\newcommand{\myauthors}[1]{
    \begingroup
    \fontsize{12pt}{14pt}\fontseries{bx}\selectfont \centering #1 \par  
    \endgroup
}

\begin{document}

\mytitle{Two Holes Make a Difference: \\Deterministic Paging with Delay and Deadlines}

\bigskip
\bigskip

\myauthors{Darya Melnyk\footnote{TU Berlin, Germany}, Linda Thelen\footnote{TU Darmstadt, Germany}, Yuyi Wang\footnote{CRRC Zhuzhou Institute \& Tengen Intelligence Institute, China}}

\bigskip
\bigskip

\begin{abstract}

In this paper, we study online paging with delay or deadlines. In online paging with delay, an online algorithm manages a cache of size $k$ and needs to serve incoming requests by adding the requested pages to the cache.
    In addition, it has the option to delay serving incoming requests at a cost that depends on the delay.
    We propose a new online algorithm, which combines a threshold rule with first-in, first-out ($\FIFO$) replacement, and show that it is $(2k+2)$-competitive for cache size $k$.
    This result corrects a previously published claim that the competitive ratio is exactly $2k+1$ regardless of the number of holes (pages outside the cache).
    However, in the case of one hole, we show that the \FIFO rule with a different threshold is indeed $(2k+1)$-competitive, giving a tight result.

    In the case of deadlines, the same analysis gives a deadline-triggered algorithm that is $(k+1)$-competitive in general and $k$-competitive if there is only one hole.
    Moreover, we show that the second hole makes paging with deadlines strictly harder than classical paging, where the competitive ratio is independent of the number of holes.
    For paging with deadlines, we prove a deterministic lower bound of $k+\tfrac12$, already for two holes.
    We show that, as a consequence, no threshold-triggered algorithm can achieve a competitive ratio less than $2k+\frac32$ in general.

\end{abstract}

\section{Introduction}

The paging problem is a fundamental problem in online computation, and it is one of the first problems studied in the competitive analysis framework~\cite{sleatorTarjanAmortizedEfficiency}.
In paging, an algorithm maintains a cache of $k$ pages from a universe of $\numPages > k$ pages.
Upon requests to pages outside the cache---the \emph{holes}---the algorithm must load them and evict others.
The problem was first introduced in its offline variant by Belady~\cite{Belady}, who was studying virtual-storage computers, which allowed algorithms to access larger amounts of data than could be stored in the physical memory of the computer.
Today, the paging problem is often presented as a way to maintain a small fast local memory---the cache---by pulling and pushing data to and from a slow large memory.
The requests are issued by the CPU directly to the cache, and one often cannot predict which page will be requested in the future.
Paging has served as the starting point for weighted~\cite{weightedPaging} and learning-augmented~\cite{MLCachingLykouris, MLCachingRohatgi} variants, and for broader models such as the $k$-server problem~\cite{CompAlgoOnlProblems} and metrical task systems~\cite{metricalTaskSystem}.

Classical paging assumes that every request needs to be served immediately.
In practice, however, many requests can wait: background jobs, prefetches, or batch queries can tolerate some latency.
As a result, requests can be batched, reducing the overall cost.
This is captured by the paging with delay problem, first introduced by Azar et al.~\cite{Azar21OSD} as a special case of online service with delay, who proved a competitive ratio of $\O(k)$.
Here, an algorithm decides not only \emph{which} page to evict but also \emph{when} a page is worth fetching.
A first exact analysis of this problem was presented in~\cite{krneticMelnykWangWattenhofer20}. The authors proved a lower bound of $2k+1$ for the competitive ratio of any algorithm and claimed to derive a tight upper bound for the problem.
However, the claimed upper bound may only hold when $k+1$ pages are considered (the case with one hole).
The example below shows that the $\LRU$ algorithm from~\cite{krneticMelnykWangWattenhofer20} fails to be $(2k+1)$-competitive on $k+2$ pages, so that the best remaining upper bound is the competitive ratio of $6k$ from \cite{Azar21OSD}.

\begin{example}
    Consider the $\LRU$ rule from~\cite{krneticMelnykWangWattenhofer20}: if the adversary requests a page that is not in the cache already, the algorithm waits until the cost of delaying the request reaches the threshold $\thresholdValue\coloneq\tfrac{k+1}k$, and then fetches it.
    To serve the page, the algorithm evicts the least recently used page from the cache.
    \Cref{fig:counterexample:teaser} presents a sample input sequence, where the competitive ratio of the algorithm is strictly larger than $2k+1$, disproving the claimed upper bound from~\cite{krneticMelnykWangWattenhofer20}.

    On this input sequence, $\LRU$ fetches a new page $5$ times, incurring delay $\delta$ each time.
    Its total cost therefore is $5\cdot(1+\thresholdValue)=12.5$. \OPT does not pay any delay cost and achieves a cost of exactly $2$.
    The competitive ratio of $\LRU$ on this input sequence is $6.25 > 5 = 2k+1$.
    Observe that the sequence can be continued by swapping the outer and the inner pages, and repeating the same request sequence. For general $k$, see \Cref{thm:fifo_lower_bound}.
 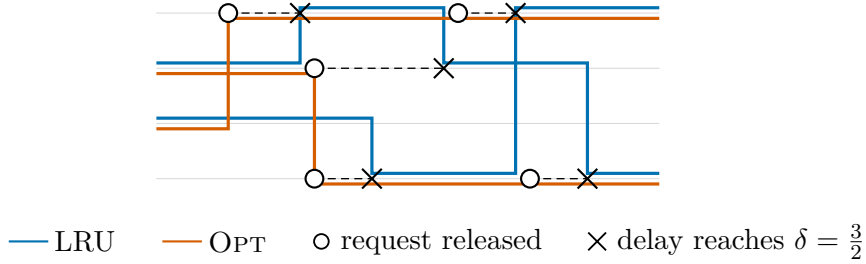
\begin{figure}[h]
        \begin{center}
           \begin{tikzpicture}[
    x=0.95cm, y=0.73cm,
    page/.style={black!15, line width=0.4pt},
    wait/.style={densely dashed, line width=0.5pt},
    rel/.style={circle, draw, fill=white, line width=0.7pt, inner sep=0pt, minimum size=5.5pt},
    crit/.style={cross out, draw, line width=0.7pt, inner sep=0pt, minimum size=6pt},
  ]
  \begin{scope}[xshift=2cm]
        \foreach \y in {0,1,2,3} {
        \draw[page] (0,\y) -- (7,\y);
        }

        \draw[alg] (0,2) -- (2,2) -- (2,3) -- (4,3) -- (4,2) -- (6,2) -- (6,0) -- (7,0);
        \draw[alg] (0,1) -- (3,1) -- (3,0) -- (5,0) -- (5,3) -- (7,3);

        \draw[opt] (0,2) -- (2.2,2) -- (2.2,0) -- (7,0);
        \draw[opt] (0,1) -- (1,1) -- (1,3) -- (7,3);

        \draw[wait] (1,3) node[release] {} -- (2,3) node[threshold]{};
        \draw[wait] (2.2,2) node[release] {} -- (4,2) node[threshold]{};
        \draw[wait] (2.2,0) node[release] {} -- (3,0) node[threshold]{};
        \draw[wait] (4.2,3) node[release] {} -- (5,3) node[threshold]{};
        \draw[wait] (5.2,0) node[release] {} -- (6,0) node[threshold]{};
  \end{scope}

  \begin{scope}[yshift=-0.9cm, every node/.append style={inner xsep=2pt}]
    \draw[alg, line width=0.8pt] (0.05,0) -- ++(0.55,0) node[anchor=west, black] (l1) {\LRU};
    \draw[opt, line width=0.8pt] ([xshift=1.4em]l1.east) -- ++(0.55,0)
      node[anchor=west, black] (l2) {\OPT};
    \node[rel] (m3) at ([xshift=1.4em+3pt]l2.east) {};
    \node[anchor=west] (l3) at ([xshift=2pt]m3.east) {request released};
    \node[crit] (m4) at ([xshift=1.4em+3pt]l3.east) {};
    \node[anchor=west] at ([xshift=2pt]m4.east) {delay reaches $\thresholdValue=\frac 3 2$};
  \end{scope}
\end{tikzpicture}
        \end{center}

        \caption{A request sequence where the $\LRU$ algorithm with delays has competitive ratio larger than $2k+1$ (illustrated for the case $k=2$, with $4$ pages). Each y-level represents a page, and the cache trajectories of \LRU and \OPT are shown in blue and orange, respectively.
    Each line represents a cache slot, and a line at the y-level of some page means that the corresponding cache slot holds the page associated with the y-level.}
        \label{fig:counterexample:teaser}
    \end{figure}
\end{example}

The results in this paper open new questions about the competitive ratio of the paging problem with delays, and suggest that there may be interesting new bounds to be discovered when two or more holes are considered, or through eviction rules other than the standard rules from classical paging.

\paragraph*{Our results}
We give a simple, nearly optimal algorithm for paging with delay, relate the remaining gap to paging with deadlines, and show that, for deadlines, two holes make paging strictly harder than classical paging.
\Cref{tab:results} summarizes all bounds.

For paging with delay, we show that a simple nonclairvoyant algorithm, which fetches a page once its accumulated delay reaches a fixed threshold $\thresholdValue$ and evicts pages in first-in, first-out order ($\thresholdValue$-$\FIFO$), is strictly $(2k+2)$-competitive.
\begin{restatable}{theorem}{thmMain}\label{thm:main}
    $1$-$\FIFO$ is strictly $(2k+2)$-competitive.
\end{restatable}
This is only an additive 1 above the lower bound of $2k+1$, which holds even for clairvoyant algorithms. On $k+1$ pages and with threshold $\frac{k+1}{k}$, the algorithm is strictly $(2k+1)$-competitive, recovering the claimed bound from~\cite{krneticMelnykWangWattenhofer20} through a new analysis.
Our analysis is tight: for every fixed threshold, $\thresholdValue$-$\FIFO$ is no better than $(2k+2)$-competitive on $k+2$ or more pages (\Cref{thm:fifo_lower_bound}).

An algorithm for paging with delay makes two decisions: when to fetch a page, which is a rent-or-buy decision, and which page to evict.
Deciding which page to evict is already reflected in the paging with deadlines problem.
For this problem, the same analysis as for \Cref{thm:main} gives an upper bound of $k+1$ in general, and a tight bound of $k$ for the case of one hole.
The latter coincides with the competitive ratio of classical paging, which does not depend on the number of holes.
For paging with deadlines, however, we prove that the problem is strictly harder with more than one hole.
\begin{restatable}{theorem}{thmDeadlinedLB}\label{thm:deadlines_lower_bound}
    Every deterministic clairvoyant online algorithm for paging with deadlines has competitive ratio at least $k+\tfrac12$ on $\numPages \geq k+2$ pages.
\end{restatable}
The reason is that, with at least two pages missing from the cache, the optimal offline algorithm can serve requests in a different order than the online algorithm does.

Lastly, we study threshold algorithms for paging with delay, which time a fetch based on a fixed threshold for the delay cost.
We show that this rule costs an additive~$k+1$ in the competitive ratio over the deadlines case: every lower bound $\rho$ for deadlines yields a lower bound of $\rho+k+1$ for threshold algorithms with delay (\Cref{prop:threshold-transfer}).
For threshold-triggered $\FIFO$, our analysis gives delay upper bounds exactly $k+1$ above the deadline upper bounds.

All proofs were verified in Lean.\footnote{Formalization: \url{https://github.com/licsth/paging-with-delay-formalization}}

\begin{table}[htbp]
    \centering
    \small
    \setlength{\tabcolsep}{10pt}
    \newcommand{\bound}[2]{\begin{tabular}[t]{@{}c@{}}#1\\[-2pt]{\scriptsize\color{gray}#2}\end{tabular}}
    \begin{tabular}{@{}lllcc@{}}
        \toprule
        Problem   & Pages               & Algorithms                & Lower bound                                                    & Upper bound                                       \\
        \midrule
        Delay     & $\numPages=k+1$     & any                       & \bound{$2k+1$}{Krneti\'c et al.~\cite{krneticMelnykWangWattenhofer20}}                                       & \bound{$\bm{2k+1}$}{\Cref{thm:fifo_k_plus_one}}        \\
        \cmidrule(l){2-5}
                  & $\numPages\geq k+2$ & any                       & \bound{$2k+1$}{Krneti\'c et al.~\cite{krneticMelnykWangWattenhofer20}}                                       & \bound{$\bm{2k+2}$}{\Cref{thm:main}}                   \\
        \addlinespace
                  &                     & threshold                 & \bound{$\bm{2k+\frac32}$}{\Cref{cor:threshold-lower}}              & \bound{$\bm{2k+2}$}{\Cref{thm:main}}                   \\
        \addlinespace
                  &                     & $\thresholdValue$-$\FIFO$ & \bound{$\bm{2k+2}$}{\Cref{thm:fifo_lower_bound}}                    & \bound{$\bm{2k+2}$}{\Cref{thm:main}}                   \\
        \midrule
        Deadlines & $\numPages=k+1$     & any                       & \bound{$k$}{Sleator and Tarjan~\cite{sleatorTarjanAmortizedEfficiency}} & \bound{$\bm k$}{\Cref{thm:fifo_deadlines_k_plus_one}} \\
        \cmidrule(l){2-5}
                  & $\numPages\geq k+2$ & any                       & \bound{$\bm{k+\frac12}$}{\Cref{thm:deadlines_lower_bound}}         & \bound{$\bm{k+1}$}{\Cref{thm:fifo_deadlines}}          \\
        \bottomrule
    \end{tabular}
    \caption{Deterministic bounds for the competitive ratios of paging with delay and with deadlines.
    Results from this paper are in bold.
        Lower bounds hold for clairvoyant algorithms and non-strict competitive guarantees; upper bounds are strict and attained by nonclairvoyant algorithms.
        The row for threshold algorithms concerns the best algorithm within that class.
        The $\thresholdValue$-$\FIFO$ lower bound holds for every threshold $\thresholdValue>0$, and its upper bound uses $\thresholdValue=1$.}
    \label{tab:results}
\end{table}

\paragraph*{Related work}

Online problems with delay were first considered in the definition of the online TCP acknowledgment problem, where sending acknowledgments of messages can be delayed and aggregated to reduce the number of sent messages, while the delay results in latency costs~\cite{onlineTCPjournal, Dooly1998TCPDA}.
The idea of explicitly adding delays to online problems was proposed by Emek et al.~\cite{onlineMatching}, who presented the problem of min-cost perfect matching with delays, showing that delaying a decision may help to match nodes at a lower cost.

The paging problem with delay was introduced by Azar et al.\ in~\cite{Azar21OSD} as a special case of online service with delay problems.
The authors present an $\O(kd^4)$-competitive clairvoyant algorithm on hierarchically separated trees (HSTs) of depth $d$.
By using a randomized embedding of general metrics into HSTs, they also present an $\O(k\log^5 n)$-competitive algorithm on general metrics.
On star metrics, which correspond to weighted paging with delay, their algorithm is $\O(k)$-competitive.
Krneti\'c et al.~\cite{krneticMelnykWangWattenhofer20} prove a $2k+1$ clairvoyant lower bound and claim a $2k+1$ nonclairvoyant upper bound. In this work, we show that the latter upper bound only holds when one hole is allowed.

For a single server in general $n$-point metric spaces, Touitou~\cite{Touitou23ImprovedDeterministic} presented an upper bound of $\O(\log n)$ on the competitive ratio of clairvoyant online service with delay.
In the nonclairvoyant case, Touitou~\cite{Touitou25NonclairvoyantOSD} later presented an upper bound of $\O(\sqrt n \log n)$.
In the special case of uniform line metrics, Bienkowski et al.~\cite{Bienkowski18OSDLine} showed an upper bound of $\O(\log n)$ without clairvoyance.

Online service with deadlines contains the $k$-server problem as a special case, where each request must be served immediately.
Thus, the clairvoyant lower bound of $k$ on the competitive ratio~\cite{sleatorTarjanAmortizedEfficiency} carries over to the paging problem with deadlines, already on $k+1$ pages.

Randomized versions of these problems have also been considered in the literature.
Gupta et al.~\cite{Gupta22CachingWithTimeWindows} give a clairvoyant upper bound of $\O(\log n \log k)$ on the weighted paging problem with delay, leaving a gap of a factor $\O(\log n)$ to the $\Omega(\log k)$ lower bound from classical paging~\cite{Fiat91CompetitivePaging}.
For classical (weighted) paging, the bound $\Theta(\log k)$ is tight~\cite{weightedPaging}.
For arbitrary metrics, Gupta et al.~\cite{Gupta22HittingSet} present a clairvoyant $\O(\mathrm{poly}\log (n,\Delta))$-competitive algorithm, where $\Delta$ is the aspect ratio of the metric space, i.e., the ratio of the largest to the smallest inter-point distance.

\section{Model and notation}\label{sec:model}

Let $\pageUniverse$ denote the set of all pages, where $|\pageUniverse| = \numPages$.
An instance $\sigma$ of \emph{paging with delay} is a sequence of requests $r=(\requestPage{r}, \arrival{r}, \delay{r})$, where $\requestPage{r} \in \pageUniverse$ is the requested page, $\arrival{r}\geq 0$ is the arrival time of the request, and $\delay{r}: [0,\infty) \to [0,\infty)$ is its delay cost function.
We require that $\delay{r}$ is continuous, nondecreasing and unbounded, and has $\delay{r}(0)=0$.
An algorithm holds a cache with up to $k$ pages at any given time and must, for each request~$r$, have $\requestPage{r}$ cached at some time in $[\arrival{r}, \infty)$. To this end, the algorithm may alter its cache instantaneously at any time, and may go through multiple cache configurations at one instant of time.
We say that a request is \emph{served} at the first time $\serviceTime{r}\geq \arrival{r}$ when its page is in the algorithm's cache.

An algorithm is \emph{online} if its actions at time $t$ depend only on the requests with arrival time at most $t$.
If its actions additionally do not depend on the trajectory of the delay cost curves after time $t$, we say that the algorithm is \emph{nonclairvoyant}.
Otherwise, the algorithm is \emph{clairvoyant}.
Further, we call an algorithm \emph{offline} if it has full knowledge of the request sequence $\sigma$ before the start of the algorithm.

We assume that the online and offline algorithms both start with the same initial cache~$\initialCache \subseteq \pageUniverse$ of size $k$. If a page corresponding to a request is already in the cache of the algorithm (cache hit), this page can be served without incurring costs. Fetching a page that was not in the cache into the cache (upon a cache miss) costs the algorithm one unit of cost, while evicting a page is free.
We write $\cache{\ALG}$ and $\cache{\OPT}$ for the current caches of the online algorithm and of the offline optimum, respectively.
The pages outside a full cache are called its \emph{holes}.
When a request $r$ is served at time $\serviceTime{r}$, the algorithm additionally pays $\delay{r}(\serviceTime{r}-\arrival{r})$ in delay cost.

An algorithm is \emph{$\rho$-competitive} if, for every input sequence $\sigma$, its total cost
\[
\ALG(\sigma)
    \coloneq
    \text{number of page fetches}
    +
    \sum_r \delay{r}(\serviceTime{r}-\arrival{r})
\]
satisfies $\ALG(\sigma)\leq \rho\cdot \OPT(\sigma) + c$, where $\OPT(\sigma)$ is the minimum cost of any (offline) algorithm on $\sigma$, and $c$ is a constant that is independent of the input sequence.
We say that the algorithm is \emph{strictly competitive} if the additive constant $c$ is $0$.
We denote \OPT's total number of fetches by $\offlineFetchCost$, and its total delay cost by $\offlineDelayCost$, so that $\OPT(\sigma)=\offlineFetchCost+\offlineDelayCost$.

\paragraph*{Threshold algorithms}
\label{sec:threshold-definition}

For a page $v$ outside the cache, the accumulated delay of its pending requests at time $t$ is
\[
    \delay{v}(t)\coloneq \sum_{\substack{r:\,\requestPage{r}=v\\r\text{ is pending at time }t}}\delay{r}(t-\arrival{r}).
\]
\begin{definition}\label{def:threshold}
    A deterministic algorithm for paging with delay is a \emph{threshold algorithm} with fixed threshold $\thresholdValue>0$ if it fetches a page $v$ exactly when $\delay{v}(t)$ reaches $\thresholdValue$ for the first time.
    It serves all pending requests for the fetched page; its replacement rule specifies which page to evict when the cache is full.
\end{definition}
We exclude the case $\thresholdValue=0$ from the definition, since it gives an unbounded competitive ratio.

We call each threshold-triggered fetch, together with the service of its pending requests, a \emph{payment}.
We write $\paymentCount$ for the number of payments, $\paymentPage{i}$ for the page fetched at payment~$i$,~$\paymentTime{i}$ for the time of the payment, and $\paymentRequests{i}$ for the requests served with that payment.
These symbols refer to the particular execution and threshold under discussion.

\paragraph*{Paging with deadlines}

In paging with deadlines, each request has a fixed deadline $\deadline{r}$, requiring that the request is served at $\serviceTime{r}\leq \deadline{r}$.
Before the deadline, the request incurs no delay, so that an algorithm's total cost is only the number of page fetches.

Paging with deadlines can be modeled as a special case of paging with delay, where the delay cost functions are zero until the deadline and ``jump'' (using a ramp) to a very large value immediately after.

In paging with deadlines, an algorithm is nonclairvoyant if its actions at time $t$ do not depend on deadlines that expire after $t$.

\paragraph*{Timing convention}
At each time, we use the convention that events happen in the following order:
First, new requests arrive and their information is revealed.
Then, the offline optimum performs its cache updates.
Afterwards, the online algorithm moves, serving all requests on visited pages.
Finally, deadlines expire (in the deadlines case), so that requests can be served exactly at the time of their deadline.

\section{FIFO upper bounds for paging with delay}\label{sec:fifo}
\subsection{The \texorpdfstring{$\thresholdValue$}{δ}-\FIFO algorithm}

We combine the threshold rule with threshold $\thresholdValue$ and the standard \emph{first-in, first-out} ($\FIFO$) cache replacement strategy.
The algorithm keeps track of the pages in the cache by ordering them in a queue by the time they were fetched into the cache. The initial $\FIFO$ queue order is arbitrary but fixed.
A cache hit does not change the ordering in the queue.
Upon a cache miss, the page at the front of the queue---the page that entered the cache earliest---is evicted, and the newly fetched page is inserted at the back of the queue.
For $\thresholdValue>0$, continuity and unboundedness of the delay functions ensure that pending requests eventually trigger a payment at exactly the threshold.
The $\thresholdValue$-\FIFO algorithm is nonclairvoyant for every choice of $\thresholdValue$.

At each payment, the pending requests being served have accumulated exactly $\thresholdValue$ units of delay cost.
Together with the unit fetch cost, every payment therefore costs $1+\thresholdValue$.

\begin{observation}\label{obs:alg-cost}
    The total cost of $\thresholdValue$-\FIFO is $(1+\thresholdValue)\paymentCount$.
\end{observation}

\subsection{Upper bounds for the competitive ratio of \FIFO}\label{sec:fifo-upper-bounds}

In this section, we prove two upper bounds. The first bound is $2k+2$ for paging with delay for any $\numPages>k$. The second is $2k+1$ for the same problem under the assumption~$\numPages=k+1$, showing that the algorithm is optimal in this special case.
We prove these bounds by dividing the time into payment windows. We then use a potential function to upper bound the competitive ratio for three possible cases.
We keep the threshold $\thresholdValue>0$ arbitrary until choosing it in the two competitive-ratio theorems.
The improvement on $k+1$ pages comes from a guaranteed potential decrease at certain offline fetches.

We may assume that $\OPT$ evicts a page only when fetching another page into a full cache: retaining a page until its slot is needed cannot increase either fetch cost or delay cost.
Thus, both caches remain full.

\paragraph*{Payment windows}

We compare the online and offline algorithms by dividing the time into payment windows. A \emph{payment window} is defined with respect to the online algorithm. A payment window $\paymentWindow{i}$ ends at payment time $\paymentTime{i}$.
It starts at time $0$ (included) if $\paymentPage{i}$ has never been evicted; otherwise, it starts immediately after the last time $\FIFO$ evicted $\paymentPage{i}$.
In the first case, $\paymentPage{i}\notin \initialCache$, i.e., a page initially in the cache must be evicted before it can trigger a payment.
The window includes payment $i$ and excludes the eviction at its start.

\begin{lemma}[Properties of payment windows]\label{lem:windows}
    For every payment $i$, the following properties hold:
    \begin{enumerate}
        \item Page $\paymentPage{i}$ is outside the cache of the online algorithm throughout $\paymentWindow{i}$ up to its final fetch, and every request in $\paymentRequests{i}$ arrives during $\paymentWindow{i}$.
        \item The requests served at the payment satisfy
              \[
                  \sum_{r\in \paymentRequests{i}}\delay{r}(\paymentTime{i}-\arrival{r})=\thresholdValue.
              \]
        \item For any fixed page, its payment windows are pairwise disjoint.
    \end{enumerate}
\end{lemma}

\begin{proof}
    If $\paymentPage{i}$ has been evicted before payment $i$, it stays outside the online cache from its last eviction until that payment.
    Otherwise, it has been absent since initialization.
    Any earlier request for $\paymentPage{i}$ was served either on arrival while the page was cached or at its preceding fetch, so every request in $\paymentRequests{i}$ arrives in $\paymentWindow{i}$.
    The threshold rule gives the second assertion.
    Finally, between two payments for the same page, that page must first be fetched and then evicted again.
    Thus, the later window starts strictly after the earlier window ends, in the fixed event order.
\end{proof}

On $k+1$ pages, $\paymentPage{i}$ is the unique page outside the online cache before payment $i$.
For~$i>1$, it was evicted at payment $i-1$, so its payment window runs from immediately after that payment through payment $i$; the first window starts at initialization.
On more than~$k+1$ pages, windows for different pages may overlap.

\paragraph*{The rank potential}
\label{sec:rank-potential}

We define the potential function using the ordering of the pages in the queue. We assign the pages in the cache of the online algorithm the ranks $1,\ldots,k$ according to the $\FIFO$ queue order. Here, rank $1$ corresponds to the page next in line to be evicted, and rank $k$ corresponds to the most recently fetched page.
Suppressing the time argument in the cache notation, define
\begin{equation}\label{eq:rank-potential}
    \potential=\sum_{q\in \cache{\ALG}\setminus \cache{\OPT}}\rank{q},
    \qquad
    0\leq\potential\leq \frac{k(k+1)}2.
\end{equation}
Up to an additive constant, this is analogous to the \FIFO potential in the paging problem without delays used in~\cite{feuerstein1998memory}.
The caches of $\thresholdValue$-\FIFO and the offline algorithm initially equal $\initialCache$, so $\potentialAt{0}=0$.
For the potential value $\finalPotential$ after both schedules have served every request, it holds that $\finalPotential\geq\potentialAt{0}=0$.

\begin{lemma}[Potential changes]\label{lem:potential-changes}
Throughout the execution, the potential changes as follows:
    \begin{enumerate}[(i)]
        \item\label{eq:potential-off} At an offline fetch, $\Delta\potential\leq k$.
        \item\label{eq:potential-off-hole} At an offline fetch, if the evicted page is outside the online cache, then $\Delta\potential\leq0$; on~$\numPages=k+1$ pages, the stronger bound $\Delta\potential\leq-1$ holds.
        \item\label{eq:potential-online} At online payment $i$, $\Delta\potential\leq-1$ if $\paymentPage{i}\in\cache{\OPT}$ immediately before the payment, and $\Delta\potential\leq k$ otherwise.
    \end{enumerate}
\end{lemma}

\begin{proof}
    We prove each item separately.
    An offline fetch adds the rank of the evicted page to the potential if that page is in the online cache, and removes the rank of the fetched page if that page is in the online cache.
    The added contribution is at most $k$, and the removed contribution is nonnegative; thus (\ref{eq:potential-off}) holds.

    To prove (\ref{eq:potential-off-hole}), note that if the evicted page is outside the online cache, no contribution is added.
    On $k+1$ pages, it is then the unique page outside the online cache, so the fetched page $q$ belongs to the online cache, and its contribution $\rank{q}\geq1$ is removed.
    With more than $k+1$ pages, both the evicted and fetched pages may be outside the online cache, in which case the potential change is zero.

    Finally, for (\ref{eq:potential-online}), for online payment $i$, write $m=|\cache{\ALG}\cap\cache{\OPT}|$ immediately before the payment.
    Each of the $k-m$ pages in the online cache but outside the offline cache loses one unit of rank contribution: the page of rank $1$ is evicted, and the ranks of the other pages decrease by one.
    The newly fetched page contributes $k$ precisely when it is outside the offline cache.
    Hence,
    \begin{equation}\label{eq:general-fifo-potential}
        \Delta\potential=m-k\,\mathds{1}_{\{\paymentPage{i}\in \cache{\OPT}\}}.
    \end{equation}
    If $\paymentPage{i}\in\cache{\OPT}$, then $m\leq k-1$, since the offline cache also contains $\paymentPage{i}$, which was absent from the online cache.
    Thus, the potential change is $m-k\leq-1$; otherwise, it is $m\leq k$.
\end{proof}

\paragraph*{Charging the payments}

The next lemma separates the charges to offline fetches and offline delay.
Its two bounds differ only in the coefficient of the fetch cost; choosing the threshold will balance the two coefficients.

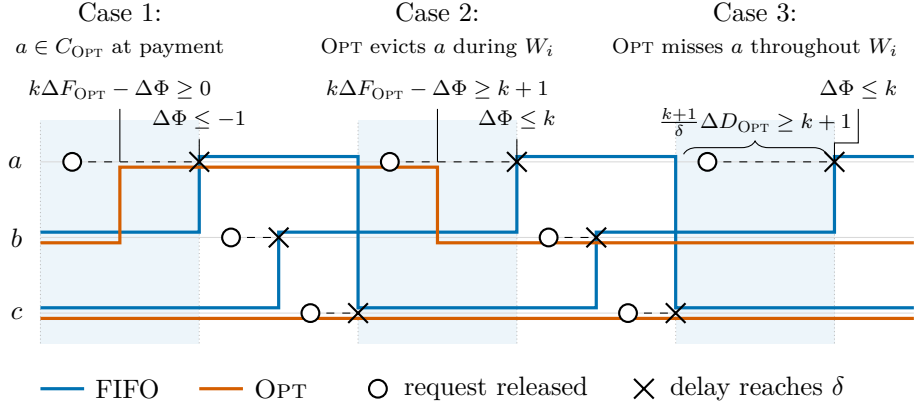
\begin{figure}[tbp]
    \centering
    \begin{tikzpicture}[x=1.05cm,y=1cm,font=\small]
        \foreach \l/\r in {0/2,4/6,8/10} {
                \fill[blue!7] (\l,-0.4) rectangle (\r,2.55);
                \draw[black!25,densely dotted] (\l,-0.4) -- (\l,2.55);
                \draw[black!25,densely dotted] (\r,-0.4) -- (\r,2.55);
            }
        \foreach \y/\p in {2/a,1/b,0/c} {
                \draw[black!15] (0,\y) -- (11,\y);
                \node[left] at (-0.1,\y) {$\p$};
            }
        \node[above, align=center] at (1,3.25) {Case 1:\\\scriptsize$a \in \cache{\OPT}$ at payment};
        \node[above,align=center] at (5,3.25) {Case 2: \\\scriptsize \OPT evicts $a$ during $\paymentWindow{i}$};
        \node[above,align=center] at (9,3.25) {Case 3:\\\scriptsize \OPT misses $a$ throughout $\paymentWindow{i}$};

        \draw[alg] (0,1) -- (2,1) -- (2,2) -- (4,2) -- (4,0)
        -- (7,0) -- (7,1) -- (10,1) -- (10,2) -- (11,2);
        \draw[alg] (0,0) -- (3,0) -- (3,1) -- (6,1) -- (6,2)
        -- (8,2) -- (8,0) -- (11,0);
        \draw[opt] (0,1) -- (1,1) -- (1,2) -- (5,2) -- (5,1)
        -- (10.6,1) -- (11,1);
        \draw[opt] (0,0) -- (11,0);

        \foreach \a/\t/\y in {0.4/2/2,2.4/3/1,3.4/4/0,
                4.4/6/2,6.4/7/1,7.4/8/0,8.4/10/2} {
                \draw[dashed] (\a,\y) node[release] {} -- (\t,\y) node[threshold] {};
            }

        \node[above] at (1,2.7) {\scriptsize$k\Delta \offlineFetchCost-\Delta \potential \geq 0$};
        \draw (1,2) -- (1,2.7);
        \node[above] at (2,2.3) {\scriptsize$\Delta \potential \leq -1$};
        \draw (2,2) -- (2,2.3);
        \node[above] at (5,2.7) {\scriptsize$k\Delta \offlineFetchCost-\Delta \potential \geq k+1$};
        \draw (5,2) -- (5,2.7);
        \node[above] at (6,2.3) {\scriptsize$\Delta \potential \leq k$};
        \draw (6,2) -- (6,2.3);

        \draw (10,2) -- (10,2.25) -- (10.3,2.25) -- (10.3,2.7) node[above] {\scriptsize$\Delta \potential \leq k$};

        \draw[decorate,decoration={brace,amplitude=5pt}] (8.1,2.1) -- (9.9,2.1)
        node[midway,above,yshift=3pt] {\scriptsize$\frac{k+1}{\thresholdValue}\Delta \offlineDelayCost \geq k+1$};

        \draw[alg, yshift=-2pt] (0,-1) -- (0.55,-1) node[right,text=black] {\FIFO};
        \draw[opt, yshift=2pt] (2,-1) -- (2.55,-1) node[right,text=black] {\OPT};
        \node[release] at (4.25,-1) {};
        \node[right] at (4.45,-1) {request released};
        \node[threshold] at (7.6,-1) {};
        \node[right] at (7.8,-1) {delay reaches $\thresholdValue$};
    \end{tikzpicture}
    \caption{The three charging cases for page $a$ with $k=2$ and $\numPages=3$.
        Each horizontal level represents a page, a colored line is a cache slot, and a vertical step $x \to y$ means evicting $x$ and fetching $y$.
        Blue lines trace the pages in the cache of \FIFO; orange lines show an offline comparator \OPT.
        The shaded windows are payment windows $\paymentWindow{i}$ for page $a$.
        During each window, the quantity~$k\offlineFetchCost+\frac{k+1}{\thresholdValue}\offlineDelayCost-\potential$ increases by at least $1$, which accounts for the online payment. This is the bound for $\numPages=k+1$ pages.}
    \label{fig:payment-cases}
\end{figure}

\begin{lemma}[Payment accounting]\label{lem:payment-accounting}
    For every threshold $\thresholdValue>0$, the number of payments of~$\thresholdValue$-$\FIFO$ satisfies
    \begin{equation}\label{eq:general-payment-accounting}
        \paymentCount+\finalPotential-\potentialAt{0}\leq(k+1)\offlineFetchCost+\frac{k+1}{\thresholdValue}\offlineDelayCost.
    \end{equation}
    On $\numPages=k+1$ pages, the stronger bound
    \begin{equation}\label{eq:k-plus-one-payment-accounting}
        \paymentCount+\finalPotential-\potentialAt{0}\leq k\offlineFetchCost+\frac{k+1}{\thresholdValue}\offlineDelayCost
    \end{equation}
    holds.
\end{lemma}

\begin{proof}
    We prove the lemma statement by showing that it holds for each payment window. In particular, to prove~\eqref{eq:general-payment-accounting}, we prove that, during a payment window $\paymentWindow{i}$,
    \[
        1 \leq (k+1)\Delta_i \offlineFetchCost+\frac{k+1}{\thresholdValue}\Delta_i \offlineDelayCost-\Delta_i \potential
    \]
    holds.
    Here, $\Delta_i \offlineFetchCost$, $\Delta_i \offlineDelayCost$ and $\Delta_i \potential$ denote the changes in $\offlineFetchCost$, $\offlineDelayCost$ and $\potential$ during the payment window related to $\paymentPage{i}$, i.e., the number of offline evictions of $\paymentPage{i}$, the delay cost \OPT accumulates for pending requests on $\paymentPage{i}$, and the potential change at payment $i$ and offline evictions of $\paymentPage{i}$.
    We prove a similar bound for $\numPages=k+1$ pages by proving~\eqref{eq:k-plus-one-payment-accounting} for each payment window $\paymentWindow{i}$:
    \[
        1 \leq k\Delta_i \offlineFetchCost+\frac{k+1}{\thresholdValue}\Delta_i \offlineDelayCost-\Delta_i \potential\,.
    \]

    Lastly, we show that for every offline eviction (in particular, those not accounted for in any payment window), the offline cost covers the corresponding potential change, i.e.,
    \begin{equation}
        \label{eq:non_accounted_offline_fetches}
        0\leq k \Delta \offlineFetchCost-\Delta \potential\,.
    \end{equation}
    By Lemma~\ref{lem:windows}, payment windows for the same page are disjoint.
    Summing over all payment windows and pages and adding \eqref{eq:non_accounted_offline_fetches} for all remaining offline fetches accounts for all offline actions and online payments exactly once, giving the claimed bounds.

    To prove \eqref{eq:non_accounted_offline_fetches}, observe that $k\Delta \offlineFetchCost - \Delta \potential = k-\Delta \potential \geq 0$, since $\Delta \potential \leq k$ by Lemma~\ref{lem:potential-changes}(\ref{eq:potential-off}).
    For the remaining bounds, fix some payment $i$.
    We distinguish three cases, illustrated in Figure~\ref{fig:payment-cases}.

    \emph{Case 1: The requested page is cached by $\OPT$ at $\paymentTime{i}$.}
    At the online payment, Lemma~\ref{lem:potential-changes}(\ref{eq:potential-online}) gives $-\Delta\potential\geq 1$, so the potential decrease pays for this online payment.
    The potential changes at offline actions are covered by the offline cost using \eqref{eq:non_accounted_offline_fetches}.
    Summing this gives
    \[
        1 \leq k\Delta_i \offlineFetchCost-\Delta_i \potential\leq k\Delta_i \offlineFetchCost+\frac{k+1}{\thresholdValue}\Delta_i \offlineDelayCost-\Delta_i \potential\,.
    \]

    \emph{Case 2: The requested page was cached by $\OPT$ during the window, but is no longer cached.}
    Then $\OPT$ evicted $\paymentPage{i}$ during $\paymentWindow{i}$.
    At that eviction, $\paymentPage{i}$ was outside the online cache; hence, $\Delta \potential \leq 0$ in general, and $\Delta \potential \leq -1$ on $k+1$ pages by Lemma~\ref{lem:potential-changes}(\ref{eq:potential-off-hole}).
    At the online payment, the potential increases by at most $k$.
    Summing these bounds and \eqref{eq:non_accounted_offline_fetches} for the remaining offline evictions in window $\paymentWindow{i}$, we get
    \[
        (k+1)\Delta_i \offlineFetchCost-\Delta_i \potential\geq (k+1)-0-k= 1
    \]
    in general, and
    \[
        k\Delta_i \offlineFetchCost-\Delta_i \potential\geq k+1-k= 1
    \]
    for $k+1$ pages.
    As before, the change in offline delay cost is nonnegative, and the claimed bounds follow.

    \emph{Case 3: The requested page is never cached by $\OPT$ during the window.}
    Every request in~$\paymentRequests{i}$ arrives in $\paymentWindow{i}$ and remains pending for both schedules until payment $i$.
    By Lemma~\ref{lem:windows}, it follows that
    \[
        \frac{k+1}{\thresholdValue}\Delta_i \offlineDelayCost \geq \frac{k+1}{\thresholdValue}
        \sum_{r\in \paymentRequests{i}}\delay{r}(\paymentTime{i}-\arrival{r})=k+1.
    \]
    At the online payment, the potential increases by at most $k$, and potential changes at offline evictions can be covered by offline fetch cost using \eqref{eq:non_accounted_offline_fetches}.
    It follows that
    \[
        k\Delta_i \offlineFetchCost+\frac{k+1}{\thresholdValue}\Delta_i \offlineDelayCost-\Delta_i \potential \geq k+1-k=1\,,
    \]
    completing the proof.
\end{proof}

\paragraph*{Competitive guarantees}

Here, we show the main result of this section, stating that the $\thresholdValue$-\FIFO algorithm is $(2k+2)$-competitive when $\thresholdValue$ is chosen appropriately.

\thmMain*

\begin{proof}
    Choose $\thresholdValue=1$ so that the fetch and delay coefficients in~\eqref{eq:general-payment-accounting} are both $k+1$.
    Since $\potentialAt{0}=0$ and $\finalPotential\geq0$, Lemma~\ref{lem:payment-accounting} gives
    \[
        \paymentCount\leq(k+1)(\offlineFetchCost+\offlineDelayCost)=(k+1)\OPT.
    \]
    Observation~\ref{obs:alg-cost} then yields
    \[
        \ALG=(1+\thresholdValue)\paymentCount=2\paymentCount\leq(2k+2)\OPT.\qedhere
    \]
\end{proof}

The following theorem shows that the algorithm is optimal when $\numPages=k+1$.

\begin{theorem}
    \label{thm:fifo_k_plus_one}
    For $k\geq 1$ and $\numPages=k+1$, $\frac{k+1}{k}$-$\FIFO$ is strictly $(2k+1)$-competitive.
\end{theorem}

\begin{proof}
    Choose $\thresholdValue=(k+1)/k$ so that the fetch and delay coefficients in~\eqref{eq:k-plus-one-payment-accounting} are both $k$.
    Since $\potentialAt{0}=0$ and $\finalPotential\geq0$, Lemma~\ref{lem:payment-accounting} gives
    \[
        \paymentCount\leq k(\offlineFetchCost+\offlineDelayCost)=k\OPT.
    \]
    Observation~\ref{obs:alg-cost} then yields
    \[
        \ALG=(1+\thresholdValue)\paymentCount=\frac{2k+1}{k}\paymentCount\leq(2k+1)\OPT.\qedhere
    \]
\end{proof}

\section{FIFO upper bounds for paging with deadlines}\label{sec:deadline-upper}
\paragraph*{The deadline-triggered \FIFO algorithm}
In this model, we define the \emph{deadline-triggered $\FIFO$} algorithm.
This algorithm fetches a noncached page when one of its pending requests reaches its deadline and applies the same replacement rule as in Section~\ref{sec:fifo}.
The fetch also serves every other pending request for that page.
We use $\paymentCount$ to denote the number of payments, $\paymentPage{i}$ for the page fetched, $\paymentTime{i}$ for the payment time, and $\paymentRequests{i}$ for the requests served with the payment, just as we did for payments with delay. Since the requests do not incur any delay cost before the deadline, the cost of the deadline-triggered $\FIFO$ algorithm is $\ALG=\paymentCount$.

\paragraph*{Upper bounds for the competitive ratio of deadline-triggered \FIFO}
We again differentiate between the general case $\numPages\geq k+1$ and the special case $\numPages=k+1$.
Both upper bounds follow from the charging argument of Section~\ref{sec:fifo-upper-bounds} once we observe that deadlines rule out its third case.

\begin{lemma}[Deadline accounting]\label{lem:deadline-accounting}
    The number of fetches of deadline-triggered $\FIFO$ satisfies
    \[
        \paymentCount+\finalPotential-\potentialAt{0}\leq(k+1)\offlineFetchCost.
    \]
    On $\numPages=k+1$ pages, the stronger bound $\paymentCount+\finalPotential-\potentialAt{0}\leq k\offlineFetchCost$ holds.
\end{lemma}
\begin{proof}
    The payment windows, the rank potential, items~1 and~3 of Lemma~\ref{lem:windows}, and Lemma~\ref{lem:potential-changes} depend only on the replacement rule and the cache contents, not on what triggers a fetch, so they apply verbatim.
    Fix a payment $i$.
    Some request $r\in\paymentRequests{i}$ has deadline $\deadline{r}=\paymentTime{i}$, and it arrives during $\paymentWindow{i}$ by Lemma~\ref{lem:windows}.
    Any feasible offline schedule serves $r$ at some time in~$[\arrival{r},\paymentTime{i}]\subseteq\paymentWindow{i}$, so $\OPT$ caches $\paymentPage{i}$ at some point during $\paymentWindow{i}$.
    Hence, Case~3 in the proof of Lemma~\ref{lem:payment-accounting} cannot occur.
    Cases~1 and~2 use only offline fetches and potential changes, and give
    \[
        1\leq(k+1)\Delta_i\offlineFetchCost-\Delta_i\potential,
        \qquad\text{and}\qquad
        1\leq k\Delta_i\offlineFetchCost-\Delta_i\potential\ \text{ on }\numPages=k+1\text{ pages}.
    \]
    Summing over all payments and adding~\eqref{eq:non_accounted_offline_fetches} for the remaining offline fetches, exactly as in the proof of Lemma~\ref{lem:payment-accounting}, proves both bounds.
\end{proof}

Finally, we can show the competitiveness of deadline-triggered $\FIFO$.

\begin{theorem}\label{thm:fifo_deadlines_k_plus_one}
    On $\numPages=k+1$ pages, deadline-triggered $\FIFO$ is strictly $k$-competitive, and no deterministic online algorithm for paging with deadlines has competitive ratio less than $k$.
\end{theorem}
\begin{proof}
    For the upper bound, Lemma~\ref{lem:deadline-accounting} with $\potentialAt{0}=0$, $\finalPotential\geq0$, $\ALG=\paymentCount$ and $\OPT=\offlineFetchCost$ gives $\ALG\leq k\OPT$.
    For the lower bound, classical paging is the special case $\deadline{r}=\arrival{r}$ for all requests, and no deterministic paging algorithm has competitive ratio less than $k$, already on $k+1$ pages~\cite{sleatorTarjanAmortizedEfficiency}.
\end{proof}

\begin{theorem}\label{thm:fifo_deadlines}
    On $\numPages\geq k+1$ pages, deadline-triggered $\FIFO$ is strictly $(k+1)$-competitive.
\end{theorem}
\begin{proof}
    Lemma~\ref{lem:deadline-accounting} with $\finalPotential\geq\potentialAt{0}=0$, $\ALG=\paymentCount$ and $\OPT=\offlineFetchCost$ gives ${\ALG\leq(k+1)\OPT}$.
\end{proof}

Together with Theorem~\ref{thm:deadlines_lower_bound}, the deterministic competitive ratio of paging with deadlines on $\numPages\geq k+2$ pages lies in $[k+\tfrac12,k+1]$.

\section{A deterministic lower bound for paging with deadlines}\label{sec:deadline-lower}

\thmDeadlinedLB*

To prove the theorem, we present a construction with $k+2$ pages and assume that the adversary always requests a page that is not in the cache of the online algorithm. Observe that such a construction would also work with $\numPages>k+2$ pages, if the adversary does not issue requests on $\numPages-k-2$ of the pages not in $\initialCache$. This is because an online algorithm that visits pages that have not been requested incurs additional moving costs and can thus only have higher costs. The cost of the online algorithm is at least the number of requests. Our goal is to find an offline schedule which pays a cost of at most $1$ every time an online algorithm pays about $k+1/2$. To analyze the costs of the offline algorithm, we define an offline budget which describes the costs an offline algorithm should stay within. With respect to this budget, we divide the request sequence into phases. We show that there exists a ``cheap'' schedule that allows the offline algorithm to have a bounded cost of at most~$1$ in each phase, while the online algorithm must incur a cost of $k$ or $k+1$. In the following, we describe the full construction.

\subsection{The offline choices and the adversary}
Without loss of generality, assume that \ALG moves only at times when a request's deadline expires.
Postponing a move until the next deadline does not increase the cost of the algorithm, and serves every request in time.

We assume that the adversary always requests pages that are currently not in the online cache, while finding an appropriate schedule for the offline algorithm. In particular, it maintains a budget for the offline algorithm. The adversary holds a set of so-called cheap offline schedules, which can shrink with each request.
Once no cheap schedule remains, the offline budget is increased by one unit, and the set of cheap schedules is refilled.
We call the sequence of requests between two budget increases a \emph{phase}.
We will show that a completed phase has at least $k$ requests, and at least every other phase has $k+1$ requests.
The competitive ratio is then shown using the property that each request increases \ALG's cost by at least one, while \OPT's cost increases by at most one per phase.

For each phase, the adversary chooses a distinguished page $h$ that is requested before the phase starts and whose deadline expires only at the end of the phase.
We call this request the \emph{held} request and write $\deadline{h}$ for its deadline.

We define the current offline \emph{budget}, denoted $\offlineBudget$, as the cost we want an offline algorithm to stay within. We define \emph{cheap schedules} as schedules that start in the configuration $\initialCache$, serve all requests so far whose deadline has passed, stay within the current offline budget, and are at the current time in a configuration with holes at the held page~$h$ and at one other page $z$.
We use $\candidates$ to denote the set of pages $z$ such that some schedule currently having holes~$h$ and $z$ is guaranteed to be cheap; thus $\candidates\subseteq \pageUniverse\setminus\{h\}$.
We maintain the following invariants for the possible current configurations, with at most one \emph{starred} member of $\candidates$ allowed to violate the third guarantee:
\begin{enumerate}
    \item\label{lb_inv_1} For every $z\neq h$, some schedule that currently has holes $h$ and $z$ has a cost of at most~$\offlineBudget+1$.
    \item\label{lb_inv_2} For every $z\in \candidates$, such a schedule costs at most $\offlineBudget$.
    \item\label{lb_inv_3} For every \emph{unstarred} $z\in \candidates$, a schedule of cost at most $\offlineBudget+1$ ends with holes $h,z$ and has also served the held request.
\end{enumerate}

\paragraph*{Construction initialization}

We restrict the construction to the pages in the initial cache and two pages $u_0$ and $v_0$ that are initially not in the cache, and write $\pageUniverse$ only for these pages.
At time $0$, the adversary releases a request at $u_0$ with deadline $1$ and one at $v_0$ with deadline $2$.
At time $1$, a new phase starts with
\[
    h=v_0,\qquad \candidates=\pageUniverse\setminus\{u_0,v_0\},\qquad \offlineBudget=1,
    \qquad\text{no starred members of\ } \candidates.
\]
We assume that the considered offline schedules end at the current time.
To denote the moves of an algorithm, we write $x\to y$ for a move where a page $x$ is evicted and $y$ is fetched instead.
For each $z\ne v_0$, we describe an offline schedule that ends with holes $v_0,z$ through a sequence of moves at time $1$, maintaining the invariants:
\begin{enumerate}
    \item To end with holes at $v_0, u_0$, perform the move $p\to u_0\to p$ for any page $p$ from $\initialCache$. This move incurs a cost of two, serves the request at $u_0$ and ends with holes $v_0,u_0$ (inv.~\ref{lb_inv_1}).
    \item For each $z\in \candidates$, the move $z\to u_0$ gives the cheap schedule with $v_0,z$ missing (inv.~\ref{lb_inv_2}).
    \item For each $z\in \candidates$, the path $z\to v_0\to u_0$ costs two, serves the held request and ends with holes $v_0,z$ (inv.~\ref{lb_inv_3}).
\end{enumerate}

\paragraph*{Continuing a phase}
The goal of the adversary is to make sure that a phase continues as long as possible. To continue a phase, the adversary chooses any page $p\neq h$ not in the cache maximizing~${|\candidates \setminus\{p\}|}$.
This is possible as long as there is some remaining $y \in \candidates \setminus\{p\}$.
In this case, release a request at $p$ with deadline $\deadline{p}$ strictly between the current time and $\deadline{h}$.
We advance time to~$\deadline{p}$, update the cheap page set to the nonempty set $\candidates \setminus\{p\}$, and erase the star at $p$ if it had one.

The invariants are preserved as follows:
\begin{enumerate}
    \item For every $z\neq h$, consider the cheap schedule ending with holes $h,y$, followed by the move~$z\to y$ after $p$ has been requested, or no action if $z=y$. This serves the new request and ends with holes $h,z$ at cost at most $\offlineBudget+1$. Note that this strategy has served all previously expired requests (inv.~\ref{lb_inv_1}).
    \item For every $z\in \candidates\setminus\{p\}$, the cheap schedule ending with holes $h,z$ serves the new request at~$p$ for free, and ends with holes $h,z$ at cost at most $\offlineBudget$ (inv.~\ref{lb_inv_2}).
    \item Analogously, for every unstarred $z\in \candidates\setminus\{p\}$, the schedule of cost $\offlineBudget+1$ that serves $h$ and ends with holes $h,z$ serves the new request at $p$ for free (inv.~\ref{lb_inv_3}).
\end{enumerate}

\paragraph*{Starting a new phase}
Suppose we cannot continue the current phase.
Choose any absent page $p\neq h$.
This page must satisfy $p \in \candidates$, since the phase could not be continued.
Release a new request at $p$ before~$\deadline{h}$, with deadline $\deadline{p}=\deadline{h}+1$.
Advance the current time to $\deadline{h}$ and start a new phase with new assignments
\[
    h'=p,\quad \candidates' = \begin{cases}
        \pageUniverse\setminus\{h,p\} & \text{if }p\text{ starred},   \\
        \pageUniverse\setminus\{p\}   & \text{if }p\text{ unstarred},
    \end{cases}
    \quad \offlineBudget'=\offlineBudget+1,
    \quad h\text{ starred if }p\text{ unstarred}.
\]

Thus, if $p$ was previously starred, the phase starts with $k$ candidates, all unstarred; otherwise, it starts with $k+1$ candidates, one starred.
\begin{figure}
    \centering
    \begin{tikzpicture}[x=0.9cm,y=0.5cm,font=\small]
        \foreach \yoff/\ltitle/\lcost/\rtitle/\rcost in {
        0/{(a) New hole $z\notin\{h,p\}$}/{old cost $\leq \offlineBudget+1$}/{(b) New hole $h^*$}/{old cost $\leq \offlineBudget$},
        -6.7/{(c) New cheap hole $z\notin\{h,p\}$}/{old cost $\leq \offlineBudget$}/{(d) Extra cheap hole $h$}/{old cost $\leq \offlineBudget+1$; $h$ served}} {
        \foreach \offset/\title/\cost in {0/\ltitle/\lcost,7/\rtitle/\rcost} {
                \begin{scope}[shift={(\offset,\yoff)}]
                    \node[align=center] at (2,4.5) {\textbf{\title}\\\cost};
                    \fill[blue!7] (3,0.65) rectangle (4,3.55);
                    \draw[black!45,densely dotted] (3,0.65) -- (3,3.55);
                    \foreach \row/\page in {3/h,2/p,1/z} {
                            \draw[black!15] (0,\row) -- (4,\row);
                            \node[left] at (-0.1,\row) {$\page$};
                        }
                    \draw[dashed] (0,3) -- (3,3) node[threshold] {};
                    \draw[dashed] (1.8,2) node[release] {} -- (4,2) node[threshold] {};
                \end{scope}
            }
        }

        \draw[opt,dotted] (0,2) -- (0.9,2);
        \draw[opt] (0.9,2) -- (3,2) -- (3,3) -- (3.7,3);
        \begin{scope}[xshift=6.3cm]
            \draw[opt,dotted] (0,1) -- (0.9,1);
            \draw[opt] (0.9,1) -- (2.5,1) -- (2.5,3) -- (3,3) -- (3, 1) -- (3.7,1);
        \end{scope}
        \node at (5.5,-0.1) {At $\deadline{h}$: cost $\leq \offlineBudget+2=\offlineBudget'+1$; new held page $p$.};

        \begin{scope}[shift={(0,-6.7)}]
            \draw[opt,dotted] (0,1) -- (0.9,1);
            \draw[opt,dotted] (7,3) -- (7.3,3) -- (7.3,1) -- (7.9,1);
            \draw[opt] (0.9,1) -- (3,1) -- (3,3) -- (3.7,3);
            \begin{scope}[xshift=6.3cm]
                \draw[opt] (0.9,1) -- (3.7,1);
            \end{scope}
            \node at (5.5,-0.1) {At $\deadline{h}$: cost $\leq \offlineBudget+1=\offlineBudget'$; new held page $p$.};

            \node[release] at (0.35,-1.4) {};
            \node[right] at (0.6,-1.4) {release};
            \node[threshold] at (4.5,-1.4) {};
            \node[right] at (4.7,-1.4) {deadline};
            \fill[blue!7] (8,-1.65) rectangle (8.7,-1.15);
            \node[right] at (8.8,-1.4) {new phase};
        \end{scope}
    \end{tikzpicture}%
    \caption{Updating the offline schedules at a phase boundary.
        Top row: schedules with budget~$\offlineBudget'+1$ (inv.~\ref{lb_inv_1} and~\ref{lb_inv_3}).
        In (a), the old witness costs at most $\offlineBudget+1$ and ends with holes at $h,z$. The move~$p \to h$ extends it to a new schedule missing $p,z$ of cost $\offlineBudget'+1$ with the held page $p$ already served (inv.~\ref{lb_inv_3}).
        In (b), the old cheap schedule missing $h,p$ is extended into a schedule with cost~$\offlineBudget'+1$ missing $h,p$. Since $p$ is not served yet, $h$ is starred.
        Bottom row: cheap schedules with budget~$\offlineBudget'$ (inv.~\ref{lb_inv_2}); the previous schedules end with holes at $h,p$.
        In (c), the old cheap witness costs at most~$\offlineBudget$.
        Swapping the page $z$ for $h$ gives the new cheap schedule missing $p,z$.
        In (d), available when $p$ was unstarred, the old witness missing $p,h$ has already served $h$ at cost at most $\offlineBudget+1$.}
    \label{fig:phase-restart}
\end{figure}
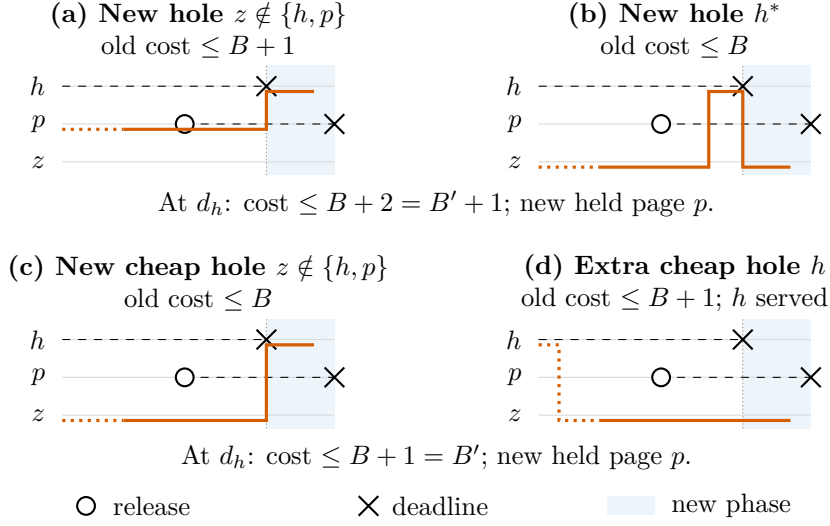

We again prove that the three invariants hold for the new phase.
Figure~\ref{fig:phase-restart}(a) and~(b) illustrate the new schedules of cost $\offlineBudget'+1$ for maintaining invariants \ref{lb_inv_1} and \ref{lb_inv_3}.
Figure~\ref{fig:phase-restart}(c) and~(d) illustrate the new cheap schedules of cost $\offlineBudget'$ (inv.~\ref{lb_inv_2}), showing how the old schedule for unstarred $p$ can be extended to give an extra cheap schedule.
\begin{enumerate}
    \item For every $z\notin \{p,h\}$, the old schedule ending with holes $h,z$ can be extended by the move~$p \to h$ to serve the old held request and end with $p,z$ missing at cost at most~$\offlineBudget+2=\offlineBudget'+1$.
          For the schedule ending with holes at $p,h$, the old cheap schedule missing~$p,h$ can be extended by a round trip $z\to h\to z$ serving $h$ with final cost~$\offlineBudget+2=\offlineBudget'+1$ (inv.~\ref{lb_inv_1}).
    \item For every $z\notin \{h,p\}$, consider the old cheap schedule missing $h,p$. At time $\deadline{h}$, move~$z \to h$. This serves the old held request and leaves $p,z$ missing, at cost at most $\offlineBudget+1=\offlineBudget'$.
          If~$h\in \candidates'$, then $p$ was previously unstarred, meaning that there exists a schedule of cost at most $\offlineBudget+1$ that ends with holes $h,p$ and has served the old held request. This schedule can stay put, giving a cheap schedule ending with holes $h,p$ (inv.~\ref{lb_inv_2}).
    \item For each $z \notin \{h,p\}$, the old schedule ending with holes $h,z$ and cost $\offlineBudget+1$ can be extended by the move~$p \to h$ at time $\deadline{h}$ to serve the old held request and the new request at $p$, and end with $p,z$ missing at cost at most $\offlineBudget+2=\offlineBudget'+1$ (inv.~\ref{lb_inv_3}).
\end{enumerate}

\subsection{Cost analysis}

Let $N \geq 2$ and consider the input $\sigma_N$ constructed by the adversary with $N$ requests released in total.

\begin{lemma}
    \label{lem:deadlines_alg_lower_bound}
    \ALG incurs cost at least $N$ on $\sigma_N$.
\end{lemma}
\begin{proof}
    The request sequence is designed such that any requested page in this sequence is not in the cache of \ALG at the time when it is requested.
    Moreover, the construction ensures that no windows of two requests at the same page overlap, so each request must be served by a separate online move.
\end{proof}

\begin{lemma}
    \label{lem:deadlines_opt_upper_bound}
    \OPT incurs cost at most $\frac{2N+2k}{2k+1}+1$ on $\sigma_N$.
\end{lemma}

\begin{proof}
    Observe that the budget $\offlineBudget$ corresponds to the number of phases started: initialization sets it to one, and each new phase increases it by one.
    We call a phase \emph{short} if it starts with~$|\candidates|=k$ and \emph{long} if it starts with $|\candidates|=k+1$.
    Observe that, if a phase is short, then it has no star, and therefore the next phase is long.
    Each step within a phase releases a request and removes at most one page from $\candidates$.
    Since a phase ends when $\candidates$ would become empty, a short phase accounts for at least $k$ requests, and a long phase accounts for at least~$k+1$ requests.
    The first $\offlineBudget-1$ phases are complete; at least $\lfloor(\offlineBudget-1)/2\rfloor$ of them are long.
    Counting also the first request of the last phase gives
    \[
        N\geq 1+k(\offlineBudget-1)+\left\lfloor\frac{\offlineBudget-1}{2}\right\rfloor
        \geq \frac{2k+1}{2}\offlineBudget-k,
        \qquad\text{hence}\qquad
        \offlineBudget\leq\frac{2N+2k}{2k+1}.
    \]
    We can now choose any cheap offline schedule as the offline algorithm. This results in a cost of at most $\offlineBudget$, and the algorithm has served every request before its deadline expired.
    To serve the held request, at most one additional move is needed, giving the claimed bound.
\end{proof}

We can finally prove the main theorem of the section:

\begin{proof}[Proof of Theorem~\ref{thm:deadlines_lower_bound}]
    Let $1\leq \rho < k+\frac{1}{2}$ and assume \ALG is $\rho$-competitive.
    By Lemmas \ref{lem:deadlines_alg_lower_bound} and \ref{lem:deadlines_opt_upper_bound}, for every $N\geq 2$, there is an input $\sigma_N$ such that
    \[
        \ALG(\sigma_N)-\rho\OPT(\sigma_N)
        \geq N\left(1-\frac{2\rho}{2k+1}\right)
        -\rho\left(1+\frac{2k}{2k+1}\right).
    \]
    This tends to infinity with $N\rightarrow\infty$, contradicting any fixed additive constant in the competitive guarantee.
\end{proof}

\section{Lower bounds for threshold algorithms}\label{sec:threshold-lower}

In this section, we focus on proving lower bounds for threshold algorithms.
We show that our analysis of the $\thresholdValue$-$\FIFO$ algorithm is tight on $\numPages\geq k+2$ pages, if the threshold $\thresholdValue$ is fixed.
Additionally, we show that any threshold algorithm has a competitive ratio of at least $2k + \tfrac32$ for online paging with delay for $\numPages\geq k+2$.

\subsection{Tightness for fixed-threshold \FIFO}

\begin{theorem}
    \label{thm:fifo_lower_bound}
    The competitive ratio of $\thresholdValue$-$\FIFO$ is at least $2k+2$ for every fixed threshold value $\thresholdValue>0$ and $\numPages\geq k+2$ pages.
\end{theorem}

\begin{proof}
    In the proof, we present repeatable adversarial runs that force the $\thresholdValue$-$\FIFO$ algorithm to incur a cost of $(1+\thresholdValue)(2k+2)$ per run, while an offline schedule pays at most $1+\thresholdValue$ per run, apart from a constant initialization and final-service cost.
    An instance of a repeatable adversarial run is depicted in Figure~\ref{fig:tightness}.

    \begin{figure}
        \begin{center}
            \begin{tikzpicture}[x=0.75cm, y=0.75cm, font=\small]

                \draw[alg] (0,0) -- (2,0) -- (2,-1) -- (6,-1) -- (6,4) -- (10,4) -- (10,3) -- (12,3);
                \draw[alg] (0,2) -- (3,2) -- (3,1) -- (7,1) -- (7,0) --(11,0) -- (11,4) -- (12,4);
                \draw[alg] (0,3) -- (4,3) -- (4,2) -- (8,2) -- (8,-1)  -- (12,-1)node[right] {\scriptsize moves next};
                \draw[alg] (0,4) -- (5,4) -- (5,3) -- (9,3) -- (9,2) -- (12,2);

                \node at (0,4) [left] {$\fixedPage{1}$};
                \node at (0,3) [left] {$\fixedPage{2}$};
                \node at (0,2) [left] {$\fixedPage{3}$};
                \node at (0,1) [left] {$a$};
                \node at (0,0) [left] {$b$};
                \node at (0,0) [left,xshift=-0.4cm] {\scriptsize moves first};
                \node at (0,-1) [left] {$c$};

                \draw[opt] (0,-1) -- (6.5,-1) -- (6.5,0) -- (12,0);
                \draw[opt] (0,2) -- (12,2);
                \draw[opt] (0,3) -- (12,3);
                \draw[opt] (0,4) -- (12,4);

                \draw[dashed] (0,1) -- (3,1) node[threshold]{};
                \draw[dashed] (0,-1) -- (2,-1) node[threshold]{};
                \draw[dashed] (2.2, 0) node[release] {} -- (7,0) node[threshold] {};
                \draw[dashed] (3.2,2) node[release] {} -- (4,2) node[threshold]{};
                \draw[dashed] (4.2,3) node[release] {} -- (5,3) node[threshold]{};
                \draw[dashed] (5.2,4) node[release] {} -- (6,4) node[threshold]{};

                \draw[dashed] (6.2, -1) node[release] {} -- (8,-1) node[threshold]{};
                \draw[dashed] (7.2, 1) node[release] {} -- (12,1);
                \draw[dashed] (8.2, 2) node[release] {} -- (9, 2) node[threshold]{};
                \draw[dashed] (9.2, 3) node[release] {} -- (10, 3) node[threshold]{};
                \draw[dashed] (10.2, 4) node[release] {} -- (11, 4) node[threshold]{};
                \draw[dashed] (11.2, 0) node[release] {} -- (12, 0);

                \draw[alg, yshift=-2pt] (0,-2) -- (0.5,-2) node[right,text=black] {\FIFO};
                \draw[opt, yshift=2pt] (2.1,-2) -- (2.6,-2) node[right,text=black] {\OPT};
                \node[release] at (4.3,-2) {};
                \node[right] at (4.5,-2) {request released};
                \node[threshold] at (8.5,-2) {};
                \node[right] at (8.7,-2) {delay reaches $\thresholdValue$};
            \end{tikzpicture}
        \end{center}
        \caption{A request sequence where $\thresholdValue$-\FIFO has competitive ratio at least $2k+2$ (illustrated is~$k=4$).
            The online cache trajectories are blue, and an offline witness is orange.
            $\FIFO$ fetches~$2k+2$ times per run, while the offline witness fetches once and leaves the request on $a$ pending.
            The start and end situations of the run are identical up to renaming $b$ and $c$, so the run can be repeated indefinitely.}
        \label{fig:tightness}
    \end{figure}
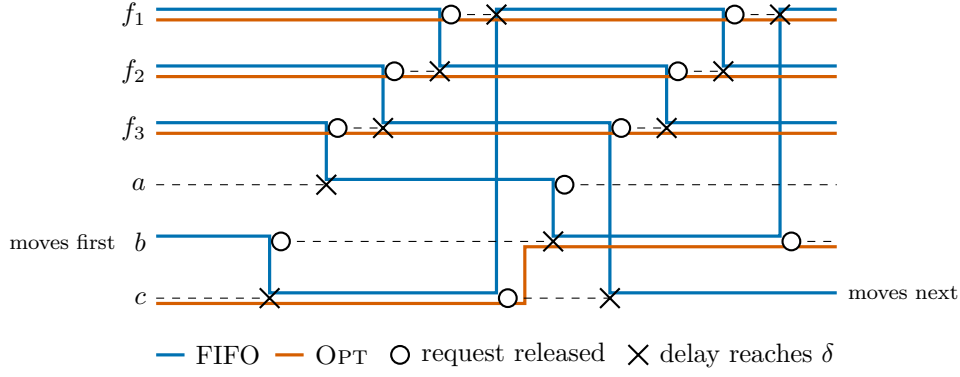

    We repeatedly issue a request right after $\FIFO$ evicts a page, and initially issue a request at each absent page used in the construction.
    A page becomes \emph{critical} when its pending request reaches delay $\thresholdValue$.
    Choose each delay function to mimic a deadline, with delay $\thresholdValue$ instead of the deadline:
    it is zero until a short continuous ramp just before its designated critical time, then constant at $\thresholdValue$ through the finite construction.
    To ensure the technical constraint that delay functions are unbounded, the delay functions may grow unboundedly after the finite construction.
    The ramps can be chosen short enough that the offline moves described below precede the relevant increases in delay.

    The request sequence is defined as follows:
    Use $k+2$ pages, labeled $\fixedPage{1}, \fixedPage{2}, \ldots, \fixedPage{k-1}$, $a$, $b$ and $c$.
    Label the initial $\FIFO$ queue, from oldest to newest, as $b,\fixedPage{k-1},\ldots,\fixedPage{1}$, leaving $a,c$ outside the cache.
    We make pages critical in the order $(c,a,\fixedPage{k-1},\dots,\fixedPage{1},b,c,\fixedPage{k-1}, \dots, \fixedPage{1})$, then relabel $b$ as $c$ and $c$ as $b$, i.e., the next run starts with $(b,a,\dots)$ in terms of the old labeling.
    It is easy to see that \FIFO faults on every request, as every run of $k+1$ requests is on distinct pages, and \FIFO covers the most recently fetched $k$ requests.

    Next, we show an upper bound on the cost of \OPT.
    To this end, we define an offline witness that first fetches $c$ and evicts $b$, at one unit of initialization cost, so its cache is~$\{\fixedPage{1},\dots,\fixedPage{k-1},c\}$.
    We describe one run of this offline witness, using it to bound $\OPT$.
    The offline witness leaves the request at page $a$ pending, accounting for at most $\thresholdValue$ delay per run.
    Note that \FIFO has the most recently fetched $k$ requests in its cache at any given time, meaning that, before $b$ becomes critical, \FIFO has $\fixedPage{1},\dots,\fixedPage{k-1}$ and $a$ in its cache.
    Therefore,~$c$ is already evicted by \FIFO, and the request on page $c$ that will turn critical after~$b$ turns critical is already issued.
    Accordingly, the offline witness performs the move $c\rightarrow b$, thereby serving the requests at $b$ and $c$ without incurring any delay cost.
    The remaining requests of the run are to pages cached by the offline witness.
    Finally, the offline witness has~$\fixedPage{1},\dots, \fixedPage{k-1}$ and $b$ in its cache, with $b$ being renamed to $c$.
    In the end, the offline witness has to fetch $a$ once to serve all outstanding requests.

    We conclude that the offline witness incurs cost at most $1+\thresholdValue$ per run, while \FIFO faults on each of the $2k+2$ requests, incurring cost $(1+\thresholdValue)(2k+2)$ ($1$ for fetching and $\thresholdValue$ for the accumulated delay).
    After $N$ runs, $\ALG\geq(1+\thresholdValue)(2k+2)N$ and $\OPT\leq(1+\thresholdValue)N+2$.
    Letting~$N$ tend to infinity proves the lower bound even when the competitive guarantee allows an additive constant.
\end{proof}

\begin{remark}
    This lower bound also holds for the \LRU algorithm, which keeps the most recently accessed pages in the cache, since requests are only to pages that are absent from the online cache, and \LRU and \FIFO coincide on such sequences.
\end{remark}

\subsection{Transferring deadline lower bounds}

We now relate deadline lower bounds to the threshold class of Definition~\ref{def:threshold}, allowing arbitrary deterministic replacement rules.

\begin{lemma}
    \label{lem:threshold-simulation}
    Let $\ALG$ be a $\rho$-competitive threshold algorithm with threshold $\thresholdValue>0$.
    Then there exists a $\frac{\rho}{1+\thresholdValue}$-competitive deterministic algorithm $\ALG'$ for paging with deadlines that is clairvoyant iff $\ALG$ is clairvoyant.
\end{lemma}

\begin{proof}
    We give the construction here only in the clairvoyant case.
    For the nonclairvoyant version, the same construction can be applied once information is revealed.

    We construct a deterministic deadline algorithm $\ALG'$ that, for an instance $\sigma'$ of paging with deadlines, simulates $\ALG$ on a delay instance $\sigma$ constructed online from $\sigma'$.
    The two algorithms start with $\initialCache$, and $\ALG'$ copies every simulated fetch and eviction.

    Let $\tau(t)$ denote the virtual time assigned to real time $t$, before processing its arrivals and deadlines.
    We set $\tau$ at arrivals and deadlines and extend it linearly in between.
    At the release of each request $r'$, we assign its counterpart $r$ an effective deadline $\deadline{r}$ and a duration~$\epsilon_r>0$.
    The intervals $[\deadline{r},\deadline{r}+\epsilon_r]$ for distinct original deadlines are disjoint and ordered in the same way as those deadlines; equal deadlines share an interval.
    We retain each interval until its original deadline has been processed, even if its requests have already been served.

    Virtual time starts at zero.
    At each real arrival or deadline time $t$, first advance virtual time to the midpoint between its current value and the start of the earliest remaining interval, or by one unit if no interval remains; this defines $\tau(t)$.
    As justified below, no fetch occurs during this advance.
    Release all requests arriving at $t$ at virtual time $\tau(t)$.
    For a deadline already announced, reuse its interval.
    Otherwise, write $r^-$ for the request preceding the new request and $r^+$ for the following one in the original deadline order of the already announced deadlines.
    If there is no preceding interval, use $\deadline{r^-}=\tau(t)$ and $\epsilon_{r^-}=0$; if there is no following interval, use $\deadline{r^+}=\deadline{r^-}+\epsilon_{r^-}+1$.
    Define the interval for the new request as the middle third of the gap between the preceding and following intervals:
    \[
        \epsilon_r=\frac{\deadline{r^+}-\deadline{r^-}-\epsilon_{r^-}}{3},
        \qquad
        \deadline{r}=\deadline{r^-}+\epsilon_{r^-}+\epsilon_r.
    \]
    The gap is positive, so $\epsilon_r>0$ and
    \[
        \deadline{r^-}+\epsilon_{r^-}
        <\deadline{r}<\deadline{r}+\epsilon_r<\deadline{r^+}.
    \]
    These choices preserve the interval order without changing any previous assignment and ensure $\arrival{r}=\tau(\arrival{r'})<\deadline{r}$.
    Define the delay function by
    \[
        \delay{r}(x)
        =\frac{\thresholdValue}{\epsilon_r}
        \max\{0,x-\deadline{r}+\arrival{r}\},\qquad x\geq0.
    \]
    It is continuous, nondecreasing, unbounded, and zero at arrival.
    It starts growing at absolute virtual time $\deadline{r}$ and reaches $\thresholdValue$ at $\deadline{r}+\epsilon_r$.
    The entire function is fixed at release using only information already available.

    After processing all arrivals at a real deadline $\deadline{r'}$, run the simulation through $\deadline{r}+\epsilon_r$, copy its cache actions at real time $\deadline{r'}$ in their simulated order, and remove the interval.
    Equal deadlines are processed together.
    Inductively, all requests with earlier deadlines have already been served, whereas requests with later deadlines have zero delay throughout this advance.
    Every request with this deadline is served by the end of its interval, since its individual delay reaches $\thresholdValue$ there and thus forces a fetch at or before that time.
    Multiple pending requests on one page may cause their aggregate delay to reach the threshold earlier, which only advances their service.
    All pending requests therefore have zero delay between successive intervals, justifying the advances to $\tau(t)$.
    This defines a feasible deterministic online algorithm $\ALG'$, including for requests arriving at their own deadlines, since arrivals are processed first.

    For the offline comparison, observe that any feasible schedule for $\sigma'$ can be transformed into a schedule for $\sigma$ with no delay, by mapping actions according to the virtual time $\tau$.
    The number of fetches is unchanged, proving $\OPT(\sigma)\leq\OPT(\sigma')$.

    Every fetch of $\ALG'$ corresponds to a fetch of $\ALG$ with associated delay cost $\thresholdValue$.
    Thus, $\ALG(\sigma)=(1+\thresholdValue)\ALG'(\sigma')$.
    Using that \ALG is $\rho$-competitive (i.e., $\ALG(\sigma)\leq\rho\OPT(\sigma)+c$ for some constant $c$), it follows that
    \[
        \ALG'(\sigma')
        =\frac{1}{1+\thresholdValue}\ALG(\sigma)
        \leq \frac{\rho}{1+\thresholdValue}\OPT(\sigma)+\frac{c}{1+\thresholdValue}
        \leq\frac{\rho}{1+\thresholdValue}\OPT(\sigma')+\frac{c}{1+\thresholdValue}.\qedhere
    \]
\end{proof}

\begin{proposition}\label{prop:threshold-transfer}
    Let $k\geq1$ and $\numPages\geq k+1$ be integers, and let $\rho\geq1$.
    If the competitive ratio of deterministic paging with deadlines on $\numPages$ pages with cache size $k$ is at least $\rho$, then no threshold algorithm for paging with delay on the same universe can have competitive ratio less than $\rho+k+1$.
\end{proposition}

\begin{proof}
    Fix a threshold algorithm $\ALG$ with threshold $\thresholdValue$.
    First, suppose that $\thresholdValue\geq(k+1)/\rho$, and assume for a contradiction that $\ALG$ has competitive ratio less than $\rho+k+1$.
    Then Lemma~\ref{lem:threshold-simulation} gives a deterministic algorithm $\ALG'$ with competitive ratio less than $\frac{\rho+k+1}{1+\thresholdValue}\leq \frac{\rho+k+1}{1+(k+1)/\rho}=\rho$, contradicting the assumed deadline lower bound.

    Now, suppose that $0<\thresholdValue<(k+1)/\rho$.
    Fix $N\geq1$ and a set of $k+1$ pages containing the initial cache.
    At each time $2i$, for $i=0,\ldots,N-1$, issue one request to the unique page in this set outside $\ALG$'s cache.
    Let its delay grow linearly from zero to $\thresholdValue$ during the next time unit, stay at $\thresholdValue$ until absolute time $2N$, and grow linearly and unboundedly thereafter.
    $\ALG$ makes exactly one fetch per round, and thus $\ALG=(1+\thresholdValue)N$.

    For an offline comparator, choose a page receiving the fewest requests.
    This page receives at most $N/(k+1)$ requests.
    An offline schedule may use at most one initial fetch to cache all other pages, keep this cache until time $2N$, and then fetch the omitted page if necessary.
    All requests to the other pages are served at arrival, and each request to the omitted page incurs delay at most $\thresholdValue$.
    Therefore,
    \[
        \OPT\leq\frac{N\thresholdValue}{k+1}+2.
    \]
    By letting $N\to\infty$, this gives a lower bound on the competitive ratio of at least
    \[
        \frac{(1+\thresholdValue)(k+1)}{\thresholdValue}
        =(k+1)\left(1+\frac{1}{\thresholdValue}\right)
        >(k+1)\left(1+\frac{\rho}{k+1}\right)
        = \rho+k+1.\qedhere
    \]
\end{proof}

The above results imply the correctness of the following theorem:

\begin{theorem}
\label{cor:threshold-lower}
    No threshold algorithm has competitive ratio less than $2k+\frac 3 2$ for online paging with delay on at least $k+2$ pages.
\end{theorem}
\begin{proof}
    Apply Proposition~\ref{prop:threshold-transfer} with $\rho=k+\tfrac12$, using Theorem~\ref{thm:deadlines_lower_bound}.
\end{proof}

\section{Discussion and future work}

One open question is to determine the optimal deterministic competitive ratio for paging with delay on at least $k+2$ pages.
Our results leave a gap between the lower bound of~$2k+1$ and the upper bound of $2k+2$.
Although the latter bound is tight for fixed-threshold $\FIFO$, it need not be optimal among all online algorithms.
Moreover, our lower bound of~$2k+\tfrac32$ for fixed-threshold algorithms shows that attaining a competitive ratio of $2k+1$, if possible, would require an algorithm outside this class.
Even for $\numPages=4$ and $k=2$, the optimal deterministic competitive ratio for paging with delay remains between $5$ and $6$.

Paging with deadlines helps isolate the cache-replacement difficulty from the additional decision of when to serve requests under delay costs.
For this variant, our deterministic bounds are $k+\tfrac12$ and $k+1$.
Understanding this gap may also help clarify the remaining gap for paging with delay.
In particular, Proposition~\ref{prop:threshold-transfer} transfers deterministic deadline lower bounds to the fixed-threshold class, with an additive increase of $k+1$.
However, this connection alone does not resolve the optimal competitive ratio for unrestricted delay algorithms.

Another important direction is to understand the benefit of randomization against an oblivious adversary.
We leave both closing the remaining gaps in the deterministic competitive ratios and investigating the power of randomization to future work.

\section*{Acknowledgments}

We would like to thank Yann Disser for useful discussions and feedback. Darya Melnyk is supported by German Research Foundation (DFG), Schwerpunktprogramm SPP 2378 - ReNO-2 (grant 511099228), 2025-2029. Linda Thelen is supported by a doctoral scholarship from the Studienstiftung des deutschen Volkes. Yuyi Wang is supported by the Science and Technology Innovation Program of Hunan Province (No.2026RC1078). 

\paragraph*{Disclosure of AI assistance}
An initial draft of the $(2k+2)$-competitive upper bound (Theorem~\ref{thm:main}) was obtained with the assistance of GPT-5.6 Sol to develop the argument.
After we proved the $(2k+1)$-competitive bound on $k+1$ pages by hand, asking GPT-5.6 Sol for potential-based proofs for the competitive ratio of classical \FIFO pointed us to the rank potential, which we then used to reprove and unify the two results.
The lower bound of $k+\frac12$ for deadline algorithms was obtained with the help of Claude Opus 5.5, after turning the basic idea of the construction of \Cref{thm:fifo_lower_bound} into a computational search over instances on four pages with $k=2$.
We then generalized the resulting rule on small instances to arbitrary universe size $\numPages=k+2$.
The proof formalizations in Lean were written entirely by Claude Opus 5.5, and the problem and algorithm definitions were written both by human authors and Claude Opus 5.5.

The authors independently checked all resulting proofs and take full responsibility for the content and accuracy of the paper, including any errors that may remain.

\bibliography{references}

@article{feuerstein1998memory,
  author    = {Esteban Feuerstein and Alberto Marchetti-Spaccamela},
  title     = {Memory Paging for Connectivity and Path Problems in Graphs},
  journal   = {Journal of Graph Algorithms and Applications},
  volume    = {2},
  number    = {6},
  pages     = {1--22},
  year      = {1998},
  doi       = {10.7155/jgaa.00009},
}

@inproceedings{krneticMelnykWangWattenhofer20,
  author    = {Predrag Krneti{\'c} and Darya Melnyk and Yuyi Wang and Roger Wattenhofer},
  title     = {The {$k$}-Server Problem with Delays on the Uniform Metric Space},
  booktitle = {Proceedings of the 31st International Symposium on Algorithms and Computation ({ISAAC})},
  series    = {LIPIcs},
  volume    = {181},
  pages     = {61:1--61:13},
  publisher = {Schloss Dagstuhl -- Leibniz-Zentrum f{\"u}r Informatik},
  year      = {2020},
  doi       = {10.4230/LIPIcs.ISAAC.2020.61},
}

@article{sleatorTarjanAmortizedEfficiency,
  author    = {Daniel D. Sleator and Robert E. Tarjan},
  title     = {Amortized Efficiency of List Update and Paging Rules},
  journal   = {Communications of the {ACM}},
  volume    = {28},
  number    = {2},
  pages     = {202--208},
  year      = {1985},
  doi       = {10.1145/2786.2793},
}

@article{Gupta22CachingWithTimeWindows,
  author    = {Anupam Gupta and Amit Kumar and Debmalya Panigrahi},
  title     = {Caching with Time Windows and Delays},
  journal   = {{SIAM} Journal on Computing},
  volume    = {51},
  number    = {4},
  pages     = {975--1017},
  year      = {2022},
  note      = {(STOC 2020)},
  doi       = {10.1137/20M1346286},
}

@inproceedings{Gupta22HittingSet,
  author    = {Anupam Gupta and Amit Kumar and Debmalya Panigrahi},
  title     = {A Hitting Set Relaxation for {$k$}-Server and an Extension to Time-Windows},
  booktitle = {Proceedings of the 62nd {IEEE} Annual Symposium on Foundations of Computer Science ({FOCS})},
  pages     = {504--515},
  publisher = {{IEEE}},
  year      = {2021},
  doi       = {10.1109/FOCS52979.2021.00057},
}

@inproceedings{Touitou25NonclairvoyantOSD,
  author    = {Noam Touitou},
  title     = {Nearly-Optimal Algorithm for Non-Clairvoyant Service with Delay},
  booktitle = {Proceedings of the 42nd International Symposium on Theoretical Aspects of Computer Science ({STACS})},
  series    = {LIPIcs},
  volume    = {327},
  pages     = {74:1--74:21},
  publisher = {Schloss Dagstuhl -- Leibniz-Zentrum f{\"u}r Informatik},
  year      = {2025},
  doi       = {10.4230/LIPIcs.STACS.2025.74},
}

@inproceedings{Touitou23ImprovedDeterministic,
  author    = {Noam Touitou},
  title     = {Improved and Deterministic Online Service with Deadlines or Delay},
  booktitle = {Proceedings of the 55th Annual {ACM} Symposium on Theory of Computing ({STOC})},
  pages     = {761--774},
  publisher = {{ACM}},
  year      = {2023},
  doi       = {10.1145/3564246.3585107},
}

@inproceedings{Bienkowski18OSDLine,
  author    = {Marcin Bienkowski and Artur Kraska and Pawe{\l} Schmidt},
  title     = {Online Service with Delay on a Line},
  booktitle = {Proceedings of the 25th International Colloquium on Structural Information and Communication Complexity ({SIROCCO})},
  series    = {LNCS},
  volume    = {11085},
  pages     = {237--248},
  publisher = {Springer},
  year      = {2018},
  doi       = {10.1007/978-3-030-01325-7\_22},
}

@article{Azar21OSD,
  author    = {Yossi Azar and Arun Ganesh and Rong Ge and Debmalya Panigrahi},
  title     = {Online Service with Delay},
  journal   = {{ACM} Transactions on Algorithms},
  volume    = {17},
  number    = {3},
  pages     = {23:1--23:31},
  year      = {2021},
  note      = {(STOC 2017)},
  doi       = {10.1145/3459925},
}

@article{Fiat91CompetitivePaging,
  author    = {Amos Fiat and Richard M. Karp and Michael Luby and Lyle A. McGeoch and Daniel D. Sleator and Neal E. Young},
  title     = {Competitive Paging Algorithms},
  journal   = {Journal of Algorithms},
  volume    = {12},
  number    = {4},
  pages     = {685--699},
  year      = {1991},
  doi       = {10.1016/0196-6774(91)90041-V},
}

@article{Belady,
  author    = {L. A. Belady},
  title     = {A Study of Replacement Algorithms for a Virtual-Storage Computer},
  journal   = {{IBM} Systems Journal},
  volume    = {5},
  number    = {2},
  pages     = {78--101},
  year      = {1966},
  doi       = {10.1147/sj.52.0078},
}

@inproceedings{CompAlgoOnlProblems,
  author    = {Mark S. Manasse and Lyle A. McGeoch and Daniel D. Sleator},
  title     = {Competitive Algorithms for On-Line Problems},
  booktitle = {Proceedings of the 20th Annual {ACM} Symposium on Theory of Computing ({STOC})},
  pages     = {322--333},
  publisher = {{ACM}},
  year      = {1988},
  doi       = {10.1145/62212.62243},
}

@article{metricalTaskSystem,
  author    = {Allan Borodin and Nathan Linial and Michael E. Saks},
  title     = {An Optimal On-Line Algorithm for Metrical Task System},
  journal   = {Journal of the {ACM}},
  volume    = {39},
  number    = {4},
  pages     = {745--763},
  year      = {1992},
  doi       = {10.1145/146585.146588},
}

@article{weightedPaging,
  author    = {Nikhil Bansal and Niv Buchbinder and Joseph Naor},
  title     = {A Primal-Dual Randomized Algorithm for Weighted Paging},
  journal   = {Journal of the {ACM}},
  volume    = {59},
  number    = {4},
  pages     = {19:1--19:24},
  year      = {2012},
  doi       = {10.1145/2339123.2339126},
}

@inproceedings{MLCachingRohatgi,
  author    = {Dhruv Rohatgi},
  title     = {Near-Optimal Bounds for Online Caching with Machine Learned Advice},
  booktitle = {Proceedings of the 31st Annual {ACM-SIAM} Symposium on Discrete Algorithms ({SODA})},
  pages     = {1834--1845},
  publisher = {{SIAM}},
  year      = {2020},
  doi       = {10.1137/1.9781611975994.112},
}

@article{MLCachingLykouris,
  author    = {Thodoris Lykouris and Sergei Vassilvitskii},
  title     = {Competitive Caching with Machine Learned Advice},
  journal   = {Journal of the {ACM}},
  volume    = {68},
  number    = {4},
  pages     = {24:1--24:25},
  year      = {2021},
  doi       = {10.1145/3447579},
}

@inproceedings{onlineMatching,
  author    = {Yuval Emek and Shay Kutten and Roger Wattenhofer},
  title     = {Online Matching: Haste Makes Waste!},
  booktitle = {Proceedings of the 48th Annual {ACM} Symposium on Theory of Computing ({STOC})},
  pages     = {333--344},
  publisher = {{ACM}},
  year      = {2016},
  doi       = {10.1145/2897518.2897557},
}

@inproceedings{Dooly1998TCPDA,
  author    = {Daniel R. Dooly and Sally A. Goldman and Stephen D. Scott},
  title     = {{TCP} Dynamic Acknowledgment Delay: Theory and Practice (Extended Abstract)},
  booktitle = {Proceedings of the 30th Annual {ACM} Symposium on Theory of Computing ({STOC})},
  pages     = {389--398},
  publisher = {{ACM}},
  year      = {1998},
  doi       = {10.1145/276698.276792},
}

@article{onlineTCPjournal,
  author    = {Daniel R. Dooly and Sally A. Goldman and Stephen D. Scott},
  title     = {On-Line Analysis of the {TCP} Acknowledgment Delay Problem},
  journal   = {Journal of the {ACM}},
  volume    = {48},
  number    = {2},
  pages     = {243--273},
  year      = {2001},
  doi       = {10.1145/375827.375843},
}

\end{document}